\documentclass[12pt]{article}
\usepackage{amsmath,amssymb}
\usepackage{graphicx}
\usepackage[dvips]{epsfig} 
\usepackage{amsmath,bm}
\usepackage[center]{subfigure}

\usepackage{times}
\usepackage{biograph}
\usepackage{cite}
 \usepackage{mathptmx}
\usepackage{latexsym}
\usepackage[latin1]{inputenc}
\usepackage[T1]{fontenc}
\usepackage{amsfonts}
\usepackage{amsthm}
\usepackage[mathscr]{eucal}
\usepackage[english,francais]{babel}
\usepackage{verbatim}
\usepackage{Times}
\usepackage{appendix}
\usepackage{color}
\usepackage{cancel}

 \newcommand{\ds}{\displaystyle}

\newtheorem{lea}{Lemma}[section]

\newtheorem{rem}{Remark}[section]
\numberwithin{equation}{section} \numberwithin{prn}{section}
\numberwithin{cor}{section} \numberwithin{thm}{section}
\numberwithin{lea}{section}

\begin{document}
\title{A short note on model selection by LASSO methods in a change-point
model}
\date{}
\author{Fuqi Chen{\thanks {University of Windsor, 401 Sunset Avenue, Windsor,
Ontario, N9B 3P4. Email: chen111n@uwindsor.ca}} \and and \quad{ }
S\'ev\'erien Nkurunziza{\thanks {University of Windsor, 401 Sunset
Avenue, Windsor, Ontario, N9B 3P4. Email: severien@uwindsor.ca}}}
\thispagestyle{empty} \selectlanguage{english} \maketitle

\maketitle\thispagestyle{empty}
\begin{abstract} In Ciuperca~(2012)~(Ciuperca.
Model selection by LASSO methods in a change-point model, {\em Stat.
Papers}, 2012;(in press)), the author considered a linear regression
model with multiple change-points occurring at unknown times. In
particular, the author studied the asymptotic properties of the
LASSO-type and of the adaptive LASSO estimators. While the
established results seem interesting, we point out some major errors
in proof of the most important result of the quoted paper.
Further, we present a corrected result and proof.
\end{abstract}

 \pagenumbering{arabic} \addtocounter{page}{0}
 \section{Introduction}
In Ciuperca~(2012), the author considered a linear regression model
with multiple change-points occurring at unknown times. In
particular, the author studied the asymptotic properties of the
LASSO-type and that of the adaptive LASSO estimators. While the
established results seem interesting, we point out a major error in
proof of one of the important result. In particular, the proof of
Part~(ii) of Lemma~3 in Ciuperca~(2011) is based on the inequality
$|a^{2}-b^{2}|\leqslant (a-b)^{2}$, which is wrong. Indeed, take
$a=2$ and $b=1$, we get $|a^{2}-b^{2}|=3>(2-1)^{2}=1$ which
contradicts the inequality used in the quoted paper.

For the sake of clarity, we use the same notation and we suppose
that the main assumptions in Ciuperca~(2012) hold. Below, we recall
these assumptions for the convenience of the reader. Namely, we
consider the following model: $Y_i=f_{\theta}(X_i)+\varepsilon_i$,
where
\begin{equation*}
f_{\theta}(X_i)=X'_{i}\phi_{1}\mathbb{\mathbb{I}}_{\left\{i<l_1\right\}}
+X'_{i}\phi_{2}\mathbf{\mathbb{I}}_{\left\{l_1\leq
i<l_2\right\}}+...+X'_{i}\phi_{K+1}\mathbf{\mathbb{I}}_{\left\{i>l_K\right\}},\quad
i=1,...n,
\end{equation*}
$\mathbb{\mathbb{I}}_{A}$ denotes the indicator function of the
event $A$, $Y_i$ denotes the response variable, $X_i$ is a
$p$-vector of regressors, $(\varepsilon_{i})_{1\leqslant i\leqslant
n}$ are the errors which are assumed to be  independent and
identically distributed (i.i.d.) random variables, $\phi_{i}\in
\Gamma\subset \mathbb{R}^{p}$, $\Gamma$ is compact, $i=1,2,\dots,K$.
The model parameters are given by $\theta=(\theta_1,\theta_2)$, with
the regression parameters $\theta_1=(\phi_1,...\phi_{k+1})$ and the
change-points $\theta_2=(l_1,...,l_k)$. In addition, we set
$\theta_{1}^0=(\phi_1^0,...\phi_{k+1}^0)$ and
$\theta_2^0=(l_1^0,...,l_k^0)$ to be the true values of $\theta_1$
and $\theta_2$, respectively. As in Ciuperca~(2012), we impose the
following conditions.

\subsection*{Main Assumptions}
\begin{description}
  \item[$\bm{(H_{1})}$] There exists two positive constants $u, c_{0} (> 0)$
such that $l_{r+1}-l_r\geqslant c_{0}[n^{u}]$,  for every
$r\in(1,...,K)$, with $l_{0}=1$ and $l_{K+1}=n$. Without loss of
generality, we consider $3/4\leqslant u\leqslant 1$, and $c_0=1$.
  \item[$(\bm{H_{2}})$] $n^{-1}\, \ds{\max_{1\leqslant i\leqslant n}}(X'_{i}X_{i})
 \xrightarrow[{n\rightarrow
\infty}]{ } \bm{0}$ and for any $r=1,...,K+1$, the matrix\\
$C_{n,r}\equiv(l_r-l_{r-1})\ds{\sum_{i=l_{r-1}+1}^{l_r}}X_{i} X'_{i}
\xrightarrow[{n\rightarrow \infty}]{ }  C_r$, where $C_r$ is a
non-negative definite matrix.
  \item[$(\bm{H_{3}})$] $\varepsilon$ is a random variable absolutely
continuous with $\textrm{E}(\varepsilon_i)=0$,
$\textrm{E}(\varepsilon_i^2)=\sigma^2$, $i=1,2,\dots, n$. 
\end{description}
We assume that $\phi_r\neq \phi_{r+1}$, $r=1,...,k$, and consider
the following penalized sum:
\begin{equation*}
S(l_1,...,l_k)=\sum_{r=1}^{k+1}\Big[\inf_{\phi_r}\sum_{i=l_{r-1}+1}^{l_r}
\Big(((Y_i-X_i\phi_r)^{2})+\frac{\lambda_{n,(l_{r-1},l_r)}}
{l_r-l_{r-1}}\sum_{u=1} ^{P}|\phi_{r,u}|^\gamma \Big)\Big],
\end{equation*}
where $\lambda_{n,(l_{r-1},l_r)}=O(l_r-l_{r-1})^{1/2}$ is the tuning
parameter and $\gamma>0$. We define the LASSO-type estimator of
$(\theta_1^0, \theta_2^0)$, say $(\hat{\theta}_1^{s},
\hat{\theta}_2^{s})$, where 
$\hat{\theta}_1^{s}=(\hat{l}_1^s,...,\hat{l}_k^s)$ and
$\hat{\theta}_2^{s}=(\hat{\phi}_1^s,...,\hat{\phi}_{k+1}^s)$, by
\begin{equation*}
\hat{\phi}_r^s=\ds{\arg\min_{\phi_r}}\sum_{i=l_{r-1}+1}^{l_r}\Big((Y_i-X_i\phi_r)^{2}
)+\frac{\lambda_{n,(l_{r-1},l_r)}}{l_r-l_{r-1}}\sum_{u=1}
^{P}|\phi_{r,u}|^\gamma \Big), \quad \forall r=1,...,k+1,
\end{equation*}
and
\begin{equation*}
\hat{\theta}_1^{s}=\ds{\arg\min_{\theta_1}}S(l_1,...,l_k).
\end{equation*}
Note that, for $\gamma=1$ and $\gamma=2$, we obtain the LASSO
estimator and ridge estimator respectively.

The rest of this paper is organized as follows.
Section~\ref{sec:mainres} gives the main result of this paper, and
in Section~\ref{sec:conclusion}. The proof of the main result is
given in the Appendix.
\section{Main result}\label{sec:mainres}
\begin{lea}\label{lema3}
Under Assumptions $(\bm{H_{2}})$, $(\bm{H_{3}})$, for all $n_1$,
$n_2\in N$, such that $n_1\geqslant n^u$, with $3/4\leq u \leq 1$,
$n_2\leq n^v$, $v<1/4$, let be the model:
\begin{eqnarray*}
Y_i&=&X_i'\phi_1^0+\epsilon_i, \; i=1,...,n_1\\
Y_i&=&X_i'\phi_2^0+\epsilon_i, \; i=n_1+1,...,n_2,
\end{eqnarray*}
with $\phi_{1}^{0}\neq \phi_{2}^{0}$. We set
$A_{n_1+n_2}^s(\phi)=\ds{\sum_{i=1}^{n_1}}\eta_{i;(0,n_1)}^s(\phi,\phi_1^0)
+\ds{\sum_{i=n_1+1}^{n_1+n_2}}\eta_{i;(n_1,n_1+n_2)}^s(\phi,\phi_2^0)$
\quad{ }
and\\
$\hat{\phi}_{n_1+n_2}^s=\ds{\arg\min_{\phi}}A_{n_1+n_2}^s(\phi)$.
Let $\delta\in (0, u-3v)$. Then,
\begin{enumerate}
  \item[(i)] $||\hat{\phi}_{n_1+n_2}^s-\phi_1^0||\leq n^{-(u-v-\delta)/2}$.
  \item[(ii)] If $\phi_2^0=\phi_1^0+\phi_3^0\,n^{-1/4}$ for some
  $\phi_3^0$, then \quad{ }
$$\ds{\sum_{i=1}^{n_1}}\eta_{i;(0,n_1)}^s(\phi_{n_1+n_2}^s,\phi_1^0)=O_p(1).$$
\end{enumerate}
\end{lea}
\begin{rem}
It should be noted that, although Part (i) of the above lemma is the
same as that of lemma~3 of Ciuperca~(2012), Part~(ii) is slightly
different. The established result holds if
$\phi_2^0=\phi_1^0+\phi_3^0\,n^{-1/4}$, while the result stated in
Ciuperca~(2012) is supposed to hold for all $\phi_2^0\neq\phi_1^0$,
but with incorrect proof. So far, we are neither  able to correct
the proof for all $\phi_2^0\neq\phi_1^0$ nor to prove that the
statement itself is wrong. Similarly, Part~(ii) of Lemmas~4~and~8
hold under the condition that
$\phi_2^0=\phi_1^0+\phi_3^0\,n^{-1/4}$.
\end{rem}
\section{Concluding Remark}\label{sec:conclusion}
In this paper, we proposed a modification of Part (ii) of Lemma~3
given in Ciuperca~(2012) for which the proof is wrong. Further, we
provided the correct proof. It should be noted that there are
several important results in the quoted paper which were established
by using Lemma~3. In particular, the quoted author used this lemma
in establishing Lemmas~4 and 8, as well as Theorems 1, 2 and 4.

\appendix
\section{Appendix}
\begin{proof}[Proof of Lemma~\ref{lema3}]
\begin{enumerate}
  \item[(i)] The proof of Part (i) is similar to that in
  Ciuperca~(2012).
\item[(ii)]  Let $Z_n(\phi)=\ds{\sum_{i=1}^{n_1}}\eta_{i}(\phi,\phi_1^0)$,
$t_n(\phi)=\ds{\sum_{i=n_1+1}^{n_1+n_2}}[(\epsilon_i-X_i'(\phi-\phi_2^0))^2-(\epsilon_i-X_i'(\phi_1-\phi_2^0))^2]$.
Then,
\begin{eqnarray*}
|t_n(\hat{\phi}_{n_1+n_2})|&=&|-2\sum_{i=n_1+1}^{n_1+n_2}\varepsilon_iX_i'(\hat{\phi}_{n_1+n_2}-\phi_1^0)+(\hat{\phi}_{n_1+n_2}-\phi_2^0)'\sum_{i=n_1+1}^{n_1+n_2}X_iX_i'(\hat{\phi}_{n_1+n_2}-\phi_2^0)\\
&&-(\phi_1^0-\phi_2^0)'\sum_{i=n_1+1}^{n_1+n_2}X_iX_i'(\phi_1^0-\phi_2^0)|.
\end{eqnarray*}
Since $||\hat{\phi}_{n_1+n_2}-\phi_1^0||\leq n^{-(u-v-\delta)/2}$,
by Cauchy-Schwarz inequality, we have
$$|\sum_{i=n_1+1}^{n_1+n_2}\varepsilon_iX_i'(\hat{\phi}_{n_1+n_2}-\phi_1^0)|\leqslant
O(n^{v/2}n^{-(u-v-\delta)/2})=o(1).$$ Further, let $\lambda_{\max}$
be the largest eigenvalue of
$\frac{1}{n_2}\ds{\sum_{i=n_1+1}^{n_1+n_2}}X_iX_i'$. Then, using the
fact that $\phi_2^0=\phi_1^0+\phi_3^0n^{-1/4}$ and Cauchy-Schwarz
inequality, we have
\begin{eqnarray*}
&&(\hat{\phi}_{n_1+n_2}-\phi_2^0)'\sum_{i=n_1+1}^{n_1+n_2}X_i X_i'
(\hat{\phi}_{n_1+n_2}-\phi_2^0)\\
&=&(\hat{\phi}_{n_1+n_2}-\phi_1^0)'n_2\frac{1}{n_2}
\sum_{i=n_1+1}^{n_1+n_2}X_iX_i'(\hat{\phi}_{n_1+n_2}-\phi_1^0)-2\phi_3^{0'} n^{-1/4}n_2\frac{1}{n_2}\sum_{i=n_1+1}^{n_1+n_2}X_i X_i'(\hat{\phi}_{n_1+n_2}-\phi_1^0)\\
&&+n^{-1/2}\phi_3^{0'} n_2\frac{1}{n_2}\sum_{i=n_1+1}^{n_1+n_2}X_i X_i'\phi_3^0\\
&\leqslant&n_2\lambda_{\max}||\hat{\phi}_{n_1+n_2}-\phi_1^0||^2+2||\phi_3^{0}|| n^{-1/4}n_2\lambda_{\max}||\hat{\phi}_{n_1+n_2}-\phi_1^0||+n^{-1/2}n_2\lambda_{\max}||\phi_3^0||^2,
\end{eqnarray*}
and then,
\begin{eqnarray*}
(\hat{\phi}_{n_1+n_2}-\phi_2^0)'\sum_{i=n_1+1}^{n_1+n_2}X_i
X_i'(\hat{\phi}_{n_1+n_2}-\phi_2^0)
=O(n^{(u-v-\delta)}n^v)+o(1)+O(n^{v-1/2})=o(1).
\end{eqnarray*}
Also, we have
\begin{eqnarray*}
&&({\phi}_{1}^0-\phi_2^0)'\sum_{i=n_1+1}^{n_1+n_2}X_iX'_i(\phi_1^0-\phi_2^0)
=n^{-1/2}\phi_3^{0'}\sum_{i=n_1+1}^{n_1+n_2}X_iX_i'\phi_3^0=O(n^{v-1/2})=o(1).
\end{eqnarray*}
Therefore, $|t_n(\hat{\phi}_{n_1+n_2})|=o_p(1)$. Further, since\\
$Z_{n}(\phi_1^0)=t_n(\phi_1^0)=0$,
$Z_{n}(\hat{\phi}_{n_1+n_2})+t_{n}(\hat{\phi}_{n_1+n_2})\leqslant
Z_{n}(\phi_1^0)+t_n(\phi_1^0)$, we have
$$0\geqslant z_{n}(\hat{\phi}_{n_1+n_2})+t_{n}(\hat{\phi}_{n_1+n_2})
\geqslant \inf_\phi Z_{n}(\phi)-|t_{n}(\hat{\phi}_{n_1+n_2})|
=\inf_\phi Z_{n}(\phi)-|o_p(1)|.$$
Hence
$$|z_{n}(\hat{\phi}_{n_1+n_2})|-|t_{n}(\hat{\phi}_{n_1+n_2})|
\leqslant |z_{n}(\hat{\phi}_{n_1+n_2})+t_{n}(\hat{\phi}_{n_1+n_2})|
\leqslant |\inf_\phi Z_{n}(\phi)|+o_p(1),$$
which implies that
$$|z_{n}(\hat{\phi}_{n_1+n_2})|\leqslant |\inf_\phi Z_{n}(\phi)|+o_p(1)
+|t_{n}(\hat{\phi}_{n_1+n_2})|=|\inf_\phi Z_{n}(\phi)|+o_p(1).$$
Let $\hat{\phi}_{n_1}=\arg\min_{\phi}Z_n(\phi)$ and $\lambda_{\max}$
be the largest eigenvalue of $n_1^{-1}\sum_{i=1}^{n_1}X_iX'_i$.
Then, by Cauchy-Schwarz inequality,
\begin{eqnarray*}
\inf_\phi
Z_{n}(\phi)\leqslant\left(\sqrt{n_1}\left\|\hat{\phi}_{n_1}-\phi_1^0\right\|\right)^2
\lambda_{\max}
+2\sqrt{n_1}\left|(\hat{\phi}_{n_1}-\phi_1^0)'n_1^{-1/2}\sum_{i=1}^{n_1}\varepsilon_iX_i\right|,
\end{eqnarray*}
and then,
\begin{eqnarray*}
\inf_\phi Z_{n}(\phi) =O_p(1)+O_p(1)O_p(1)=O_p(1), \quad{ } \mbox{
and }  \quad{ } |z_{n}(\hat{\phi}_{n_1+n_2})|=O_p(1).
\end{eqnarray*}

 Now, 
let
\begin{eqnarray*}
Z_n^s(\phi)&=&\sum_{i=1}^{n_1}\eta_{i}(\phi,\phi_1^0)+\lambda_{n;(0,n_1)}
[\sum_{k=1}^p(|\phi_{,k}|^\gamma-|\phi_{1,k}^0|^\gamma)]\\
t_n^s(\phi)&=&\sum_{i=n_1+1}^{n_1+n_2}[(\epsilon_i-X_i'(\phi-\phi_2^0))^2
-(\epsilon_i-X_i'(\phi_1-\phi_2^0))^2]\\
& & \quad{ } +\lambda_{n;(n_1,n_1+n_2)}
[\sum_{k=1}^p(|\phi_{,k}|^\gamma-|\phi_{1,k}^0|^\gamma)].
\end{eqnarray*}
 Then,
$$A_{n_1+n_2}^s(\phi)=Z_n^s(\phi)+t_n^s(\phi)-(\epsilon_i-X_i'(\phi_1-\phi_2^0))^2
+\lambda_{n;(n_1,n_1+n_2)}[\sum_{k=1}^p(|\phi_{1,k}^0|^\gamma-|\phi_{2,k}^0|^\gamma)].$$
Then
$\hat{\phi}_{n_1+n_2}^s=\arg\min_{\phi}(Z_n^s(\phi)+t_n^s(\phi))
=\arg\min_{\phi}A_{n_1+n_2}^s(\phi)$.
In addition, using the similar approach as previous, with the fact
that $||\hat{\phi}_{n_1+n_2}^s-\phi_1^0||\leqslant
n^{-(u-v-\delta)/2}$ and $\phi_2^0=\phi_1^0+\phi_3^0n^{-1/4}$, we
have
\begin{eqnarray*}
|t_{n}^s(\hat{\phi}_{n_1+n_2}^s)|&\leqslant& o(1)+\lambda_{n;(n_1,n_1+n_2)}
[\sum_{k=1}^p(|\hat{\phi}_{n_1+n_2,k}^s|^\gamma-|\phi_{1,k}^0|^\gamma)]\\
&=&o_p(1)+O(n^{v/2})O_p(||\hat{\phi}_{n_1+n_2}^s-\phi_1^0||)=O_p(n^{-(u-2v-\delta)/2})
=o_p(1).\end{eqnarray*}
Besides, $Z_n^s(\phi_1^0)=t_n^s(\phi_1^0)=0$, thus
\begin{eqnarray*}
0\geqslant\inf_\phi(Z_n^s(\phi_1^0)+t_n^s(\phi_1^0))
=Z_{n}^s(\hat{\phi}_{n_1+n_2}^s)+t_{n}^s(\hat{\phi}_{n_1+n_2}^s)
=Z_{n}^s(\hat{\phi}_{n_1+n_2}^s)-|o_p(1)|\\
\geqslant
\inf_{\phi}Z_n^s(\phi)-|o_p(1)|.
\end{eqnarray*}
Hence
$$|Z_n^s(\hat{\phi}_{n_1+n_2}^s)|\leqslant|\inf_\phi Z_n^s(\phi)|+o_p(1).$$
On the other hand, since
$$0\geqslant\inf_\phi Z_n^s(\phi)\geqslant \inf_\phi Z_n(\phi)+\lambda_{n;(0,n_1)}
\inf_\phi[\sum_{k=1}^p(|\phi_{,k}|^\gamma-|\phi_{1,k}^0|^\gamma)],$$
$$|\inf_\phi Z_n^s(\phi)|\leqslant |\inf_\phi Z_n(\phi)|+|\lambda_{n;(0,n_1)}
\inf_\phi[\sum_{k=1}^p(|\phi_{,k}|^\gamma-|\phi_{1,k}^0|^\gamma)]|.$$
Further,
$$\inf_\phi[\sum_{k=1}^p(|\phi_{,k}|^\gamma-|\phi_{1,k}^0|^\gamma)]
\leqslant \sum_{k=1}^p(|\hat{\phi}_{n_1,k}|^\gamma-|\phi_{1,k}^0|^\gamma)
=O_p(||\hat{\phi}_{n_1}-\phi_1^0||)=O_p(n_1^{-1/2}),$$
and $\inf_\phi Z_n(\phi)=O_p(1)$. It follows that $|\inf_\phi
Z_n^s(\phi)|\leqslant O_p(1)$. Hence,
$$|Z_n^s(\hat{\phi}_{n_1+n_2}^s)|\leqslant|\inf_\phi Z_n^s(\phi)|+o_p(1)=O_p(1).$$
\end{enumerate}
\end{proof}

\end{document}